\documentclass[journal,twoside,web]{ieeecolor}
\usepackage{lcsys}

\usepackage[absolute,overlay]{textpos}

\title{
Online Safety under Multiple Constraints and Input Bounds using \gatekeeper{}: \\Theory and Applications
}

\def\BibTeX{{\rm B\kern-.05em{\sc i\kern-.025em b}\kern-.08em
    T\kern-.1667em\lower.7ex\hbox{E}\kern-.125emX}}
\author{Devansh R. Agrawal, \IEEEmembership{Student Member, IEEE}, and Dimitra Panagou, \IEEEmembership{Senior Member, IEEE}
\thanks{
The authors would like to acknowledge the support of the National Science Foundation (NSF) under grant no. 2137195. Both D Agrawal ({\tt\small devansh@umich.edu}) and D Panagou ({\tt\small dpanagou@umich.edu}) are with the Robotics Department,  University of Michigan, Ann Arbor, USA.
}
\thanks{
$^\dagger$ {\tt \scriptsize \href{https://github.com/dev10110/GatekeeperFormationFlight.jl}{https://github.com/dev10110/GatekeeperFormationFlight.jl}}
}
}

\usepackage{acronym}

\usepackage{blindtext} 

\usepackage[T1]{fontenc}

\usepackage{amsmath}
\usepackage{amssymb}
\usepackage{amsfonts}
\usepackage{graphicx}
\usepackage{textcomp}
\usepackage{setspace}
\usepackage{booktabs}
\usepackage{multirow}
\usepackage[table]{xcolor} 
\usepackage{epsfig}
\usepackage{svg}
\usepackage[linesnumbered,ruled,vlined]{algorithm2e} 
\usepackage{bm}
\usepackage{tabularx}
\usepackage{placeins} 
\usepackage{siunitx} 
\usepackage{hyperref}

\usepackage{cite}

\usepackage{xurl}

\newcommand{\naturals}{\mathbb{N}}

\newcommand{\reals}{\mathbb{R}}
\newcommand{\R}{\reals}
\newcommand{\Rnonneg}{\reals_{\geq 0}}

\renewcommand{\S}{\mathbb{S}}
\newcommand{\pd}{\S_{+}}

\newcommand{\doubleto}{\rightrightarrows}
\newcommand{\st}{\mathrm{s.t.}}

\newcommand{\norm}[1]{\left\Vert #1 \right \Vert}

\newcommand{\bmat}[1]{\begin{bmatrix}#1\end{bmatrix}}

\newcommand{\argmin}[1]{\underset{\substack{#1}}{\operatorname{argmin}}}

\newcommand{\minimize}[1]{\underset{\substack{#1}}{\operatorname{minimize}}}

\newcommand{\Ccal}{\mathcal{C}}
\newcommand{\Dcal}{\mathcal{D}}

\newcommand{\Fcal}{\mathcal{F}}

\newcommand{\Lcal}{\mathcal{L}}

\newcommand{\Scal}{\mathcal{S}}
\newcommand{\Tcal}{\mathcal{T}}
\newcommand{\Ucal}{\mathcal{U}}

\newcommand{\Xcal}{\mathcal{X}}

\newcommand{\eqn}[1]{\begin{align} #1 \end{align}}
\newcommand{\neqn}[1]{\begin{align*} #1 \end{align*}}
\newcommand{\seqn}[2][]{
\begin{subequations}
#1
\begin{align} #2 \end{align}
\end{subequations}
}

\newcommand{\gatekeeper}{\texttt{gatekeeper}}

\newcommand{\nom}{{\rm nom}}
\newcommand{\can}{{\rm can}}
\newcommand{\com}{{\rm com}}
\newcommand{\back}{{\rm bak}}

\newcommand{\opt}{{\rm opt}}

\newcolumntype{g}{>{\columncolor{gray!30}}r}

\acrodef{BCH}[BCH]{Baker-Campbell-Hausdorff}
\acrodef{CBF}[CBF]{Control Barrier Function}
\acrodef{CBF-QP}[CBF-QP]{Control Barrier Function Quadratic Program}
\acrodef{CDC}[CDC]{Conference on Decision and Control}
\acrodef{CESDF}[CESDF]{Certified ESDF}
\acrodef{CLF}[CLF]{Control Lyapunov Function}
\acrodef{CVO}[C-VO]{Certified Visual Odometry}
\acrodef{DCT}[DCT]{Discrete Cosine Transform}
\acrodef{DMP}[DMP]{Distance Map Planner}
\acrodef{EKF}[EKF]{Extended Kalman Filter}
\acrodef{ESDF}[ESDF]{Euclidean Signed Distance Field}
\acrodef{EZ}[EZ]{Engagement Zone}
\acrodef{FOV}[FoV]{Field of View}
\acrodef{FPV}[FPV]{First Person View}
\acrodef{GNC}[GNC]{Graduated-Nonconvexity}
\acrodef{GP}[GP]{Gaussian Process}
\acrodef{HOCBF}[HOCBF]{Higher Order Control Barrier Function}
\acrodef{ICCBF}[ICCBF]{Input-Constrained Control Barrier Function}
\acrodef{IEEE}[IEEE]{Institute of Electrical and Electronics Engineers}
\acrodef{IMU}[IMU]{Inertial Measurement Unit}
\acrodef{ISS}[ISS]{Input-to-State}
\acrodef{KF}[KF]{Kalman Filter}
\acrodef{ML}[ML]{Machine Learning}
\acrodef{MPC}{Model Predictive Control}
\acrodef{NGPKF}[NGPKF]{Numerical Gaussian Process Kalman Filter}
\acrodef{ODE}[ODE]{Ordinary Differential Equation}
\acrodef{QP}[QP]{Quadratic Program}
\acrodef{RGBD}[RGBD]{RGB-Depth}
\acrodef{RL}[RL]{Reinforcement Learning}
\acrodef{RoS}[RoS]{Rate of Spread}
\acrodef{SDE}[SDE]{Stochastic Differential Equation}
\acrodef{SDF}[SDF]{Signed Distance Field}
\acrodef{SFC}[SFC]{Safe Flight Corridor}
\acrodef{SLAM}[SLAM]{Simultaneous Localization and Mapping}
\acrodef{SOS}[SOS]{Sum of Squares}
\acrodef{SVD}[SVD]{Singular Value Decomposition}
\acrodef{TCAC}[TCAC]{Technical Committee on Aerospace Controls}
\acrodef{TLS}[TLS]{Truncated Least Squares}
\acrodef{TSDF}[TSDF]{Truncated Signed Distance Field}
\acrodef{TSD}[TSD]{Target Spatial Distribution}
\acrodef{UAV}[UAV]{Unmanned Aerial Vehicle}
\acrodef{VIO}[VIO]{Visual Inertial Odometry}
\acrodef{VO}[VO]{Visual Odometry}
\acrodef{WLS}[WLS]{Weighted Least Squares}

\usepackage{amsthm}

\theoremstyle{plain}
\newtheorem{theorem}{Theorem}
\newtheorem{corollary}{Corollary}
\newtheorem{lemma}{Lemma}
\newtheorem{problem}{Problem}
\newtheorem{definition}{Definition}

\newtheorem{assumption}{Assumption}
\newtheorem{remark}{Remark}

\theoremstyle{remark}

\AtBeginEnvironment{definition}{\let\oldemph\emph \renewcommand{\emph}[1]{\textbf{#1}}}
\AtEndEnvironment{definition}{\let\emph\oldemph}

\usepackage{cleveref}
\crefformat{problem}{Problem~#2#1#3}
\crefformat{assumption}{Assumption~#2#1#3}

\begin{document}

\maketitle
\thispagestyle{empty}
\pagestyle{empty}

\begin{textblock*}{20cm}(1cm,1cm) 
{Author's Copy.  This paper has been accepted for publication to IEEE L-CSS 2025. }
\end{textblock*}

\begin{abstract}
This letter presents an approach to guarantee online safety of a cyber-physical system under multiple state and input constraints. 
Our proposed framework, called \gatekeeper{}, recursively guarantees the existence of an infinite-horizon trajectory that satisfies all constraints and system dynamics.
Such trajectory is constructed using a backup controller, which we define formally in this paper. \gatekeeper{} relies on a small number of verifiable assumptions, and is computationally efficient since it requires optimization over a single scalar variable. We make two primary contributions in this letter. (A)~First, we develop the theory of \gatekeeper{}: we derive a sub-optimality bound relative to a full nonlinear trajectory optimization problem, and show how this can be used in runtime to validate performance. This also informs the design of the backup controllers and sets. (B)~Second, we demonstrate in detail an application of \gatekeeper{} for multi-agent formation flight, where each Dubins agent must avoid multiple obstacles and weapons engagement zones, both of which are nonlinear, nonconvex constraints. \href{https://github.com/dev10110/GatekeeperFormationFlight.jl}{[Code]${}^\dagger$}
\end{abstract}


\section{INTRODUCTION}
\label{section:introduction}

\IEEEPARstart{I}{ncreasing} use of robotic systems in real-world applications necessitates advanced controllers that ensure safety, robustness, and effectiveness in human-machine teaming~\cite{alves2018considerations}.

This letter formalizes and builds upon our recent work on online safety verification and control~\cite{agrawal2024gatekeeper}, which introduces \gatekeeper{} as a novel algorithmic component between the planner and the controller of the autonomous system. 
To briefly illustrate the principle behind \gatekeeper{}, consider a \ac{UAV} navigating an unknown environment. The \ac{UAV} follows a nominal trajectory, generated by its planner and tracked by its controller. At each iteration, \gatekeeper{} performs two key steps: (i)~it evaluates the currently known safe set (derived from onboard sensing), and a backup set, which represents a region the \ac{UAV} can retreat to if the nominal trajectory is predicted to exit the safe set in the future; (ii)~it constructs a candidate trajectory by stitching together the nominal trajectory (up to a future time horizon) and a backup trajectory that leads safely into the backup set. The candidate is accepted if it remains within the known safe set. If so, it becomes the new committed trajectory to be tracked by the controller. Otherwise, the \ac{UAV} continues to follow the previously committed trajectory. Because a new trajectory is only committed when guaranteed to be safe, the \ac{UAV} is always lies within the safe set. Importantly, \gatekeeper{} only forward propagates candidate trajectories, making it computationally efficient.

\emph{Literature Review:} Various approaches guarantee safety of autonomous systems. \ac{MPC} offers a natural framework, but solving nonlinear/nonconvex problems online can be computationally expensive or can fail without warning~\cite{borrelli2017predictive, lopez2019dynamic, yin2025safe}. \acp{CBF} enforce safety constraints through \acp{QP}~\cite{ames2016control, garg2024advances, cohen2024safety}, though finding a valid barrier function remains challenging in the face of multiple constraints and input bounds. Reachability-based methods offer strong safety guarantees~\cite{mitchell2005time, ganai2024hamilton}, but are intractable in high-dimensions. 

Increasingly backup-based methods have become popular~\cite{tordesillas2019faster, chen2021backup, singletary2022safe, hobbs2023runtime, jung2024contingency, janwani2024learning, ko2024backup}, where a precomputed fallback is used as the system approaches unsafe conditions. Our framework extends this idea while addressing key limitations. Compared to~\cite{tordesillas2019faster} it explicitly considers nonlinear systems and nonconvex constraints. Compared to~\cite{chen2021backup, singletary2022safe, ko2024backup} instead of mixing the nominal and backup control inputs, we check when it is necessary to switch to the backup, allowing the system to follow the nominal closely. By guaranteeing safety using trajectories rather than controllers, we can enable performant backup maneuvers. \cite{hobbs2023runtime}~reviews runtime assurance methods, each variant appropriate for a different application. A common challenge is in choosing when/how to intervene, which we address by analyzing the optimality of backups.

\emph{Contributions:} One open question is, how optimal is the generated trajectory of \gatekeeper{}, and how is this affected by the choice of the backup controller and set? This letter formalizes \gatekeeper{} by deriving suboptimality bounds, and defining the formal construction of the backup controller. This also provides a framework to analyze the (sub)optimality of other safety architectures, since most methods (e.g. ~\cite{tordesillas2019faster, singletary2022safe}) do not consider/minimize the penalty on mission performance.

Second, we demonstrate the framework in a challenging multi-agent formation flight problem. Each Dubins agent must avoid multiple \acp{EZ}, both of which are nonlinear, nonconvex constraints that depend on the robot's state. This combination of tight input bounds, multiple constraints, and nonconvexity means that most modern approaches fail to guarantee safety. Furthermore, we demonstrate that our solution is computationally efficient and close-to-optimal solutions can be computed in 1\% of the time required to solve this problem using IPOPT.

\section{THEORY}

\emph{Notation: } $\naturals= \{0, 1, 2, ...\}$ is the set of natural numbers. 
$\R, \Rnonneg$ denote reals, and non-negative reals.
$\pd^n$ is the set of symmetric positive-definite matrices in $\R^{n \times n}$. The notation $\{1:N\}$ defines the set $\{1, ..., N \} \subset \naturals$. For $v\in \R^n$, $\norm{v} = \sqrt{ v^T v}$, $\norm{v}_P = \sqrt{ v^T P v}$. The set of piecewise continuous functions $w : \Tcal \to \Dcal$ are denoted by $\Lcal(\Tcal, \Dcal)$, and $\Tcal = \R$ when omitted. $\doubleto$ denotes a set-valued map, e.g. $\Scal: \Tcal \doubleto \Xcal$ means that for any $t \in \Tcal$, $\Scal(t) \subset \Xcal$ is a set.

\subsection{Preliminaries}

Consider a (possibly time-varying) dynamical system
\eqn{
\label{eqn:dynamics}
\dot x = f(t, x, u),
}
where $x \in \Xcal \subset \R^n$ is the state, and $u \in \Ucal \subset \R^m$ is the control input. The dynamics $f: \Rnonneg \times \Xcal \times \Ucal \to \R^n$ are piecewise continuous in $t$ and locally Lipschitz in $x$ and $u$. Given a feedback policy $u = \pi(t, x)$ with $\pi$ piecewise continuous in $t$ and locally Lipschitz in $x$, the closed-loop system admits a unique solution over some interval.

\begin{definition}[Trajectory]
    Let $\Tcal = [t_i, t_f) \subset \R$. A trajectory is a pair of functions $(p: \Tcal \to \Xcal , u: \Tcal \to \Ucal)$ satisfying
    \eqn{
    \dot p(t) = f(t, p(t), u(t)) \quad \forall t \in (t_i, t_f).
    }
The set of all trajectories from $(t, x) \in \R \times \Xcal$ is
\eqn{
\Phi(t, x) = \{ (p, u) :  p(t) = x \text{ and } (p, u) \text{ is a trajectory}\}.
}
\end{definition}

Let $\Scal: \R \doubleto \Xcal$ denote the (possibly time-varying) set of states satisfying constraints, e.g., a polytope $\{x : Ax \leq b\}$ or superlevel set $\{x : h(x) \geq 0\}$. The system satisfies the constraints if $x(t) \in \Scal(t) \quad \forall t \geq t_0$. In principle, one can compute the set of initial states that admit a safe trajectory:
\eqn{
\Fcal(t) = \Big \{ x & \in \Xcal : \exists (p, u) \in \Phi(t, x) \text{ satisfying} \nonumber \\
& \quad p(\tau) \in \Scal(\tau) \ \forall \tau \geq t \Big \}.
}
Computing $\Fcal(t)$ is generally intractable, as it involves solving a reachability problem. Instead, we assume a known backup set $\Ccal : \R \doubleto \Xcal$:
\begin{definition}[Backup set]
    $\Ccal: \R \doubleto \Xcal$ is a backup set if
        \eqn{
            \Ccal(t) \subset \Scal(t) \quad \forall t \in \R,
        }
    and the controller $\pi^B: \R \times \Xcal \to \Ucal$ is such that for all $t_i \in \R$ the closed-loop system $\dot x = f(t, x, \pi^B(t, x))$ satisfies 
    \eqn{
        x(t_i) \in \Ccal(t_i) \implies x(t) \in \Ccal(t) \ \forall t \geq t_i.
    }
\end{definition}

\begin{lemma}
    If $\Ccal$ is a backup set, then
    \eqn{
\Ccal(t) \subset \Fcal(t) \subset \Scal(t) \quad \forall t \in \R.
}
\end{lemma}
\begin{proof}
    If $x \in \Fcal(t)$, then by definition, $x = p(t) \in \Scal(t)$. If $x \in \Ccal(t)$, the backup controller ensures $x(\tau) \in \Ccal(\tau) \subset \Scal(\tau)$ for all $\tau \geq t$, hence $x \in \Fcal(t)$.
\end{proof}

We specify the mission objectives in terms of a desired/nominal trajectory for the system to follow. Formally, 
\begin{definition}[Nominal Trajectory]
\label{def:nominal_trajectory}
    Given state $x_k \in \Xcal$ at time $t_k \in \R$, a planner generates a nominal trajectory, 
    \neqn{
    (p^{\nom}_k , u^\nom_k) \in \Phi(t_k, x_k)
    }    
    defined over $\Tcal = [t_k, t_k + T_H]$. 
\end{definition}

We cannot directly execute the nominal trajectory since it may violate safety constraints and may not end in $\Fcal(t_k + T_H)$, risking future constraint violation.



\subsection{Problem Statement}


To address this, we seek a modified trajectory that is both safe and tracks the nominal plan. We can pose this as:
\seqn[\label{eqn:opt_prob_C}]{
\minimize{ \substack{p \in \Lcal(\Xcal), \\ u \in \Lcal(\Ucal)} } & \int_{t_k}^{t_k + T_H} L \left( t, p(t), u(t), p^{\nom}_k(t), u^{\nom}_k(t) \right) dt \label{eqn:opt_prob_C:obj} \\
\st~& \dot p = f(t, p(t), u(t)), \quad \quad \forall t \in \Tcal, \label{eqn:opt_prob_C:dyn}\\
 &p(t) \in \Scal(t), \quad \quad \quad \quad \quad \forall t \in \Tcal, \label{eqn:opt_prob_C:safety}\\
 &p(t_k) = x_k, \label{eqn:opt_prob_C:initial}\\
 &p(t_k + T_H + T_B) \in \Ccal(t_k + T_H + T_B), \label{eqn:opt_prob_C:terminal}
 }
with $\Tcal = [t_k, t_k + T_H + T_B]$. 


This is a finite-horizon optimal control problem. The terminal constraint~\eqref{eqn:opt_prob_C:terminal} ensures the trajectory ends in a backup set, which is stricter than requiring $p(t_k + T_H) \in \Fcal(t_k + T_H)$. We choose the former since $\Fcal$ is unknown. The objective~\eqref{eqn:opt_prob_C:obj} is to minimize the cost of deviating from the nominal trajectory. Note, the cost only integrates over $[t_k, t_k + T_H]$ a subset of $\Tcal$.  We make an assumption on $L$:
\begin{assumption}
\label{assumption:L}
    $L: \R \times \Xcal \times \Ucal \times \Xcal \times \Ucal \to \Rnonneg$ is positive definite about $(x_2, u_2)$:
    \neqn{
    &L(t, x_1, u_1, x_2, u_2) \geq 0, \\
    &L(t, x_1, u_1, x_2, u_2) = 0 \iff (x_1 = x_2 \text{ and } u_1 = u_2).
    }
    for all $t \in \R, x_1, x_2 \in \Xcal, u_1, u_2 \in \Ucal$.
\end{assumption}

Examples satisfying this include:
\seqn{
L_1(\cdot) &= \norm{x_1 - x_2}^2_Q + \norm{u_1 - u_2}^2_R,\\
L_2(\cdot) &= e^{-\gamma (t-t_k)}L_1(\cdot),\\
L_3(\cdot) &= \begin{cases}
    0 & \text{if } x_1 = x_2 \text{ and } u_1  = u_2,\\
    1 & \text{else}.
\end{cases} \label{eqn:L3}
}
where $Q \in \pd^n$, $R \in \pd^m$, $\gamma > 0$

The problem addressed in this paper is as follows:
\begin{problem}
    \label{problem:main}
    Design a method to solve~\eqref{eqn:opt_prob_C} under~\Cref{assumption:L} and assuming a backup set is known for the system. If solutions are suboptimal, quantify the suboptimality.
\end{problem}

In summary, the goal is to compute a \emph{committed trajectory} $(p^\com, u^\com)$ that guarantees safety over $[t_k, \infty)$ while minimizing deviation from the nominal trajectory. As~\eqref{eqn:opt_prob_C} is typically nonconvex, our focus is on efficient computation of feasible solutions and real-time suboptimality bounds. 

\subsection{Constructing candidate trajectories}

Here we describe our solution, a recursive method for constructing committed trajectories. The main result, \Cref{theorem:suboptimality}, proves feasibility and quantifies suboptimality. While the core principles of \gatekeeper{} as presented in~\cite{agrawal2024gatekeeper} remain unchanged, key differences are highlighted in \Cref{remark:differences_to_TRO}.


\gatekeeper{} triggers at discrete times $t_k \in \R$ for $k \in \naturals$. Each candidate trajectory is a concatenation of two segments: the nominal trajectory (as defined in~\Cref{def:nominal_trajectory}) and a backup trajectory, defined as follows:
\begin{definition}[Backup Trajectory] 
\label{def:backup_trajectory}
    Let $T_B \geq 0$. For any $t_s \in \R$ and $x_s \in \Xcal$, a trajectory $(p^\back, u^\back) \in \Phi(t_s, x_s)$ defined on $[t_s, \infty)$ is a backup trajectory from $(t_s, x_s)$ if 
    \eqn{
    p^\back(t_s + T_B) &\in \Ccal(t_s + T_B),
    }
    and for all $t \geq t_s + T_B$, $u^\back$ satisfies
    \eqn{
    u^\back(t) &= \pi^B(t, p^\back(t)).
    }
\end{definition}
In words, a trajectory $(p^\back,  u^\back)$ is a backup trajectory from the specified $(t_s, x_s)$, if (A)~trajectory is dynamically feasible starting from $(t_s, x_s)$ (since $(p^\back,  u^\back) \in \Phi(t, x)$), (B)~the trajectory reaches $\Ccal$ within $T_B$ seconds, and (C)~after reaching $\Ccal$, the control input corresponds to the backup controller. Recall from the definition of $\pi^B$ this ensures the trajectory remains within $\Ccal$. Thus for any backup trajectory, we have $p^\back(\tau) \in \Ccal(\tau)$ for all $\tau \geq t_s + T_B$. 

A candidate trajectory is one that switches between executing the nominal trajectory and a backup trajectory. This idea of stitching together a section of the nominal trajectory with a backup trajectory is a core principle of \gatekeeper{}.
\begin{definition}[Candidate trajectory]
\label{def:candidate_trajectory}
    Consider a system with state $x_k \in \Xcal$ at time $t_k \in \R$. Let the nominal trajectory be $(p_k^\nom, u_k^\nom) \in \Phi(t_k, x_k)$ defined over $[t_k, t_k + T_H]$. A candidate trajectory with switch time $t_s \in [t_k, t_k + T_H]$ is $(p_k^\can, u_k^\can) \in \Phi(t_k, x_k)$ defined by 
    \eqn{
    (p_k^\can(\tau), u_k^\can(\tau)) = \begin{cases}
        (p_k^\nom(\tau), u_k^\nom(\tau)) & \text{if } \tau \in [t_k, t_s),\\
        (p_k^\back(\tau), u_k^\back(\tau)) & \text{if } \tau \geq t_s,
    \end{cases}
    }
    where $(p_k^\back, u_k^\back)$ is a backup trajectory from $(t_s, p_k^\nom(t_s))$. 
\end{definition}

A candidate is \emph{valid} if it is safe over a finite horizon:
\begin{definition}[Valid] A candidate trajectory $(p_k^\can, u_k^\can) \in \Phi(t_k, x_k)$ with switch time $t_s \in \R$ is valid if
        \eqn{
        p_k^\can(t) \in \Scal(t) \quad \forall t \in [t_k, t_s + T_B],
        }
        where $T_B \geq 0$ is the horizon of the backup trajectory.
\end{definition}

This immediately leads to the following lemma:
\begin{lemma}
\label{lemma:candidate_is_safe}
    Consider a system with state $x_k \in \Xcal$ at time $t_k \in \R$. If $(p_k^\can, u_k^\can) \in \Phi(t_k, x_k)$ is a valid candidate trajectory,
    then 
    \eqn{
    p_k^\can(t) \in \Scal(t) \quad \forall t \geq t_k.
    }
\end{lemma}
\begin{proof}
    First we prove that $p_k^\can(t) \in \Scal(t)$ for all $t \geq t_k$. Since the candidate is valid, we  have $p_k^\can(t) \in \Scal(t) \ \forall t \in [t_k, t_b]$ where $t_b = t_s + T_B$. 
    Since the candidate trajectory follows the backup trajectory for all $t \in [t_s, t_b]$, it reaches the backup set at $t_b$, i.e., $p_k^\can(t_b) \in \Ccal(t_b)$. For all $t \geq t_b$ the control input matches the backup controller, and because $\pi^B$ renders $\Ccal$ forward invariant, it follows that $p_k^\can(t) \in \Ccal(t) \forall t \geq t_b$. Finally since for any backup set $\Ccal(t) \subset \Scal(t) \ \forall t \in \R$, we have $ p_k^\can(t) \in \Ccal(t) \subset \Scal(t) \quad \forall t \geq t_b$. Therefore, we have $p_k^\can(t) \in \Scal(t)$ for all $t \geq t_k$. 
\end{proof}

This proves that any valid candidate trajectory is a safe trajectory for all future time, but only requires one to check safety over a finite horizon $[t_k, t_s + T_B]$. This finite-horizon check enables practical implementation of the framework.

\subsection{Optimality of candidate trajectories}

Beyond safety, we would also like our trajectories to be optimal. Here we address: (A)~how should one construct the backup trajectory, and (B)~how should one select the best candidate trajectory, i.e., select the best switching time $t_s$?

Recall that $L(\cdot)$ is the running cost, as defined in~\eqref{eqn:opt_prob_C:obj}. The objective functional~\eqref{eqn:opt_prob_C:obj} can be split into two intervals:
\seqn{
J(p, u) &= \int_{t_k}^{t_k + T_H} L (\cdot) dt \\
&= \underbrace{\int_{t_k}^{t_s} L (\cdot) dt}_{J_1(p, u, t_s)} + \underbrace{\int_{t_s}^{t_k + T_H} L (\cdot) dt}_{J_2(p, u, t_s)}
}
where $(\cdot) = \left( t, p(t), u(t), p^{\nom}_k(t), u^{\nom}_k(t) \right)$, and $t_s \in [t_k, t_k + T_H]$ is a switch time. Then we have the following:
\seqn{
&J(p^\can, u^\can) \nonumber\\
&\quad = J_1(p^\can, u^\can, t_s) + J_2(p^\can, u^\can, t_s)\\
&\quad = J_1(p^\nom, u^\nom, t_s) + J_2(p^\back, u^\back, t_s)\\
&\quad = 0 + J_2(p^\back, u^\back, t_s)
}
where the $J_1$ term is zero, since for all $t \in [t_k, t_s]$, the candidate trajectory matches the nominal trajectory. Thus, by~\Cref{assumption:L}, the integrand is $L(\cdot) = 0$, and therefore $J_1(\cdot) = 0$. We can thus propose a solution to~\Cref{problem:main}:
\begin{theorem}
\label{theorem:suboptimality}
    Suppose at time $t_k \in \R$ the system state is $x_k \in \Xcal$. Let $(p_k^\nom, u_k^\nom) \in \Phi(t_k, x_k)$ be the nominal trajectory. Suppose $(p_k^\can, u_k^\can) \in \Phi(t_k, x_k)$ is a candidate trajectory with switch time $t_s \in [t_k, t_k + T_H]$. 
    
    If $(p_k^\can, u_k^\can)$ is a valid candidate trajectory, then
    \begin{enumerate}
        \item[(A)] the candidate trajectory is a feasible solution to~\eqref{eqn:opt_prob_C},  
        \item[(B)] the suboptimality of the candidate trajectory with respect to~\eqref{eqn:opt_prob_C} is upper-bounded by
        \footnote{The integral in~\eqref{eqn:suboptimality_bound} is over $[t_s, t_k + T_H]$, not $[t_k, t_k + T_H]$ as in~\eqref{eqn:opt_prob_C:obj}.}
        \eqn{
        \bar B = \int_{t_s}^{t_k + T_H} L(\cdot) dt \label{eqn:suboptimality_bound}
        }
        where $(\cdot) = (t, p_k^\can(t), u_k^\can(t), p^\nom(t), u^\nom(t))$. 
    \end{enumerate}
\end{theorem}

\begin{proof}
    
    \emph{Claim~A:} Since $(p_k^\can, u_k^\can) \in \Phi(t_k, x_k)$ it must satisfy~\eqref{eqn:opt_prob_C:dyn} and~\eqref{eqn:opt_prob_C:initial}. Since it is valid, it must satisfy~\eqref{eqn:opt_prob_C:safety}, and at time $t_{b} = t_s + T_B \leq t_k + T_H + T_B$ (since $t_s \in [t_k, t_k + T_H]$) the trajectory reaches $p_k^\can(t_b) \in \Ccal(t_b)$. Since for $t \geq t_b$ the candidate trajectory remains within $\Ccal$, it satisfies~\eqref{eqn:opt_prob_C:terminal}. 
    
    \emph{Claim~B:} For convenience, let $t_H = t_k + T_H$ and $t_{HB} = t_k + T_H + T_B$. Suppose $(p_k^\opt, u_k^\opt) \in \Phi(t_k, x_k)$ is the optimal solution of~\eqref{eqn:opt_prob_C}, which exists since $(p_k^\can, u_k^\can)$ is feasible. 
    
    First, notice that it is possible for $J(p_k^\opt, u_k^\opt) \geq 0$. Suppose the nominal trajectory is safe (i.e., $p_k^\nom(t) \in \Scal(t) \ \forall t \in [t_k, t_H]$) and terminates in the backup set (i.e., $p_k^\nom(t_H) \in \Ccal(t_H)$). Then, the candidate trajectory $(p', u')$ with switch time $t_s = t_H$ is valid. Notice that $J(p', u') = 0$ since for all $t \in [t_k, t_H]$ the candidate trajectory is equal to the nominal trajectory, and thus $L(\cdot) = 0$. Therefore, it is possible for $J(p, u) = 0$ in~\eqref{eqn:opt_prob_C}, and thus $J(p_k^\opt, u_k^\opt) \geq 0$.

    Second, notice that the cost of the candidate trajectory is 
    \neqn{
    J(p_k^\can, u_k^\can) &= \underbrace{J_1(p_k^\can, u_k^\can, t_s)}_{= 0} + \underbrace{J_2(p_k^\can, u_k^\can, t_s)}_{= \bar B}
    }
    where the the first term is zero since over the interval $t \in [t_k, t_s]$ the candidate trajectory equals the nominal trajectory. Therefore, 
    $0 \leq J(p_k^\opt, u_k^\opt) \leq J(p_k^\can, u_k^\can) = \bar B$,
    and thus 
    \neqn{
    J(p_k^\can, u_k^\can) - J(p_k^\opt, u_k^\opt) \leq \bar B, 
    }
    i.e., $\bar B$ is the maximum suboptimality. 
\end{proof}

\begin{corollary}
    The optimal candidate trajectory has
    \eqn{
    t_s \in \argmin{t_s' \in [t_k, t_k + T_H]} J_2(p^\back, u^\back, t_s'). \label{eqn:ts_search}
    }
    where for any $t_s' \in \R$,  $(p^\back, u^\back)$ is a backup trajectory from $(t_s', p_k^\nom(t_s'))$.
\end{corollary}

\subsection{Optimal backup trajectories}

In principle, as part of solving~\eqref{eqn:ts_search}, one could also optimize over the set of backup trajectories. For any given $t_s'$, suppose the backup trajectory solves
    \seqn[\label{eqn:opt_prob_backup}]{
    \minimize{ (p, u) \in \Phi(t_s', p_k^\nom(t_s'))} \ & J_2(p, u, t_s') \label{eqn:opt_prob_backup:obj} \\
        \st \ &p(t) \in \Scal(t), \quad \quad  \forall t \in \Tcal, \label{eqn:opt_prob_backup:safety}\\
        &p(t_s' + T_B) \in \Ccal(t_s' + T_B), 
        \label{eqn:opt_prob_backup:terminal}
    }
    where $\Tcal = [t_s', t_k + T_B]$. In this case, we can conclude:
\begin{lemma}
\label{lemma:optimal_backup}
    Let $(p_k^\nom, u_k^\nom) \in \Phi(t_k, x_k)$ be the nominal trajectory. A candidate trajectory $(p^\can, u^\can) \in \Phi(t_k, x_k)$ with switch time $t_s$ is an optimal solution of~\eqref{eqn:opt_prob_C} if 
    \begin{itemize}
        \item[(A)] the candidate trajectory is valid,
        \item[(B)] the backup trajectory is the solution of~\eqref{eqn:opt_prob_backup}, and
        \item[(C)] the switch time $t_s$ is chosen according to~\eqref{eqn:ts_search}.
    \end{itemize}  
\end{lemma}
\begin{proof}
    Notice that a candidate trajectory satisfying conditions (A, B, C) is the solution of the problem
    \seqn{
    \minimize{ \substack{p \in \Lcal(\Xcal),  u \in \Lcal(\Ucal) \\ t_s \in [t_k, t_k + T_H]} } \ & J_1(p, u, t_s) + J_2(p, u, t_s)\\
        \st \ & \eqref{eqn:opt_prob_C:dyn}, \eqref{eqn:opt_prob_C:safety},  \eqref{eqn:opt_prob_C:initial}, \eqref{eqn:opt_prob_C:terminal}
    }
    which is equivalent to problem~\eqref{eqn:opt_prob_C}, except with the additional variable $t_s$. This additional variable does not affect the feasibility or optimality of the problem, therefore these optimization problems (and solutions) are equivalent.
\end{proof}


\begin{remark}
The key insight from~\Cref{lemma:optimal_backup} is that if an optimal backup trajectory is known, solving~\eqref{eqn:ts_search} yields the optimal solution to~\eqref{eqn:opt_prob_C}, but without requiring one to solve a trajectory optimization problem.  \gatekeeper{} is particularly useful when feasible (but not necessarily optimal) backup trajectories can be efficiently generated. In such cases, \eqref{eqn:ts_search} --- a scalar line search over a bounded interval --- yields a suboptimal solution to~\eqref{eqn:opt_prob_C}, with a suboptimality bound given by~\eqref{eqn:suboptimality_bound}.
\end{remark}

\subsection{The \gatekeeper{} architecture}

\gatekeeper{} is an intermediary module between the planner and the low-level controller. It uses the planner's nominal trajectory to construct a committed trajectory that is defined for all future time, guaranteed to be safe, and minimally deviates from the nominal. At each planning step $k$ at time $t_k$ and state $x_k$, given a nominal trajectory $(p_k^\nom, u_k^\nom)$ \gatekeeper{} constructs candidate trajectories for each switch time $t_s \in [t_k, t_k + T_H]$ (see~\Cref{def:candidate_trajectory}). Valid candidates are checked, and if any exist, the one minimizing the cost in~\eqref{eqn:ts_search} is selected; otherwise, the committed trajectory remains unchanged. This guarantees safety for all future time if a valid committed exists initially.

While the optimal backup could be found by solving~\eqref{eqn:opt_prob_backup}, this may be computationally expensive. Instead, efficient suboptimal backups can still yield good performance, with online-computable suboptimality bounds given in~\Cref{theorem:suboptimality}, enabling runtime monitoring of mission progress.

\begin{remark} \label{remark:differences_to_TRO}
The \gatekeeper{} framework was first introduced in~\cite{agrawal2024gatekeeper}, but this work includes several key extensions:
\begin{itemize}
\item Whereas~\cite{agrawal2024gatekeeper} selected the switch time $t_s$ to maximize the validity duration of the nominal trajectory (i.e., minimizing the backup duration), we generalize this by allowing arbitrary cost functions in~\eqref{eqn:ts_search}.
\item Disturbances are handled in~\cite{agrawal2024gatekeeper} but are omitted here for clarity; the analysis can be extended to include them.
\item We introduce a formal optimality framework, including sufficient conditions for optimality (\Cref{lemma:optimal_backup}) and suboptimality bounds (\Cref{theorem:suboptimality}), which were not considered in~\cite{agrawal2024gatekeeper} or prior works like~\cite{tordesillas2019faster, singletary2022safe}.
\end{itemize}
\end{remark}

\section{APPLICATION}
\label{section:applications}

We demonstrate the proposed architecture to a multi-agent formation flight problem, where agents must safely navigate through a domain with multiple \acp{EZ}. 

\subsubsection*{Problem Setup}

Consider a team of $N_A$ agents:
\eqn{
\dot x_i &= \bmat{v_i \cos x_{i, 3},  & v_i \sin x_{i, 3}, & \omega_i}^T 
}
where $x_i \in \Xcal \subset \R^3$ and $u_i = [v_i, \omega_i] \in \Ucal \subset \R^2$. The state defines 2D position and heading; inputs are bounded: $v_i \in [0.8, 1.0]$, $\omega_i \in [-10.0, 10.0]$. Units are normalized ($\rm LU/TU = 1$). Agents must avoid $N_Z$ \acp{EZ}:
\eqn{
x_i \in \Scal = \{ x : h_j(x) \geq 0, \ \forall j \in \{1,\dots,N_Z\} \},
}
where $h_j$ is defined in~\cite{chapman2025engagement}. We omit its expression for brevity, and as it is not critical to the algorithm.
The leader’s trajectory $(p^L, u^L)$ is precomputed using a Dubins-based RRT* method~\cite{wolek2024sampling}, generating a safe path in $\sim$13 seconds.
Followers track offset curves from the leader’s path, using a forward-propagated controller to ensure feasibility. Inter-agent collision constraints are not considered.

This problem is challenging since it includes (A)~multiple safety constraints, (B)~tight input bounds, and (C)~complicated $h_j$ expressions that are difficult to handle analytically.

\subsubsection*{Applying \gatekeeper{}}

\begin{figure*}
    \centering
    \includegraphics[width=\linewidth]{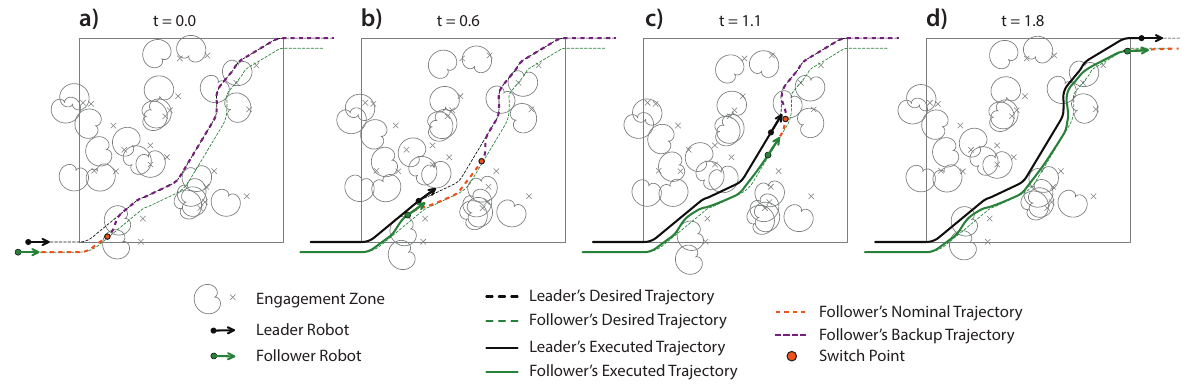}
    \caption{Formation flight using \gatekeeper{}. The leader follows an RRT* path (black dashed); followers track offset curves. The committed trajectory (green) includes nominal (orange dashed) and backup (purple dashed) trajectories.}
    \label{fig:gk}
\end{figure*}

We define the backup set as the leader’s trajectory:
\eqn{
\Ccal = \{ x \in \Xcal: \exists \tau \in [t_0, t_f] \  \st \  p^L(\tau) = x \}.
}

This set is forward invariant under the backup controller:
\eqn{
\pi^B(t, x) = u^L(t - t' + \tau), \quad \forall t \geq t',
}
where at time $t'$ the robot joins the leaders path: $x(t') = p^L(\tau)$. Thus, $\pi^B$ keeps the agent on $\Ccal$ for all $t \geq t'$. To compute backup trajectories, we use Dubins shortest paths (using~\cite{shkel2001classification, DubinsJL}) to a discrete set of points in $\Ccal$, rejecting unsafe paths and selecting the shortest safe one. Candidate trajectories are constructed by minimizing the cost in~\eqref{eqn:ts_search}, using the norm $\norm{x - x_{\rm nom}}_Q$ with $Q=\rm{diag}(1, 1, 0)$.

\Cref{fig:gk} depicts committed trajectories. They follow the nominal until nearing an \ac{EZ}, then switch to tracking the leader’s path. Since committed trajectories are updated, agents can rejoin the nominal trajectory after passing an \ac{EZ}.

\subsubsection*{Results}

\Cref{fig:results} summarizes the results.  We compare \gatekeeper{} against (A)~\acp{CBF-QP} and (B)~nonlinear trajectory optimization.  
The \ac{CBF-QP} minimizes deviation from nominal input while satisfying $
L_fh_j(x_i) + L_gh_j(x_i) u  \geq -\alpha(h_j(x_i))$, for all $j$
subject to input bounds. This may become infeasible due to multiple constraints or invalidity of $h_j$ as a \ac{CBF}. When infeasible, we re-solve a slacked version of the \acs{QP}, minimizing the norm of the slack. 
The nonlinear trajectory optimization uses IPOPT~\cite{pulsipher2022unifying}, which replans every 0.2~TU for a 0.5~TU horizon, initialized using the nominal trajectory.

\begin{figure*}
    \centering
    \includegraphics[width=\linewidth]{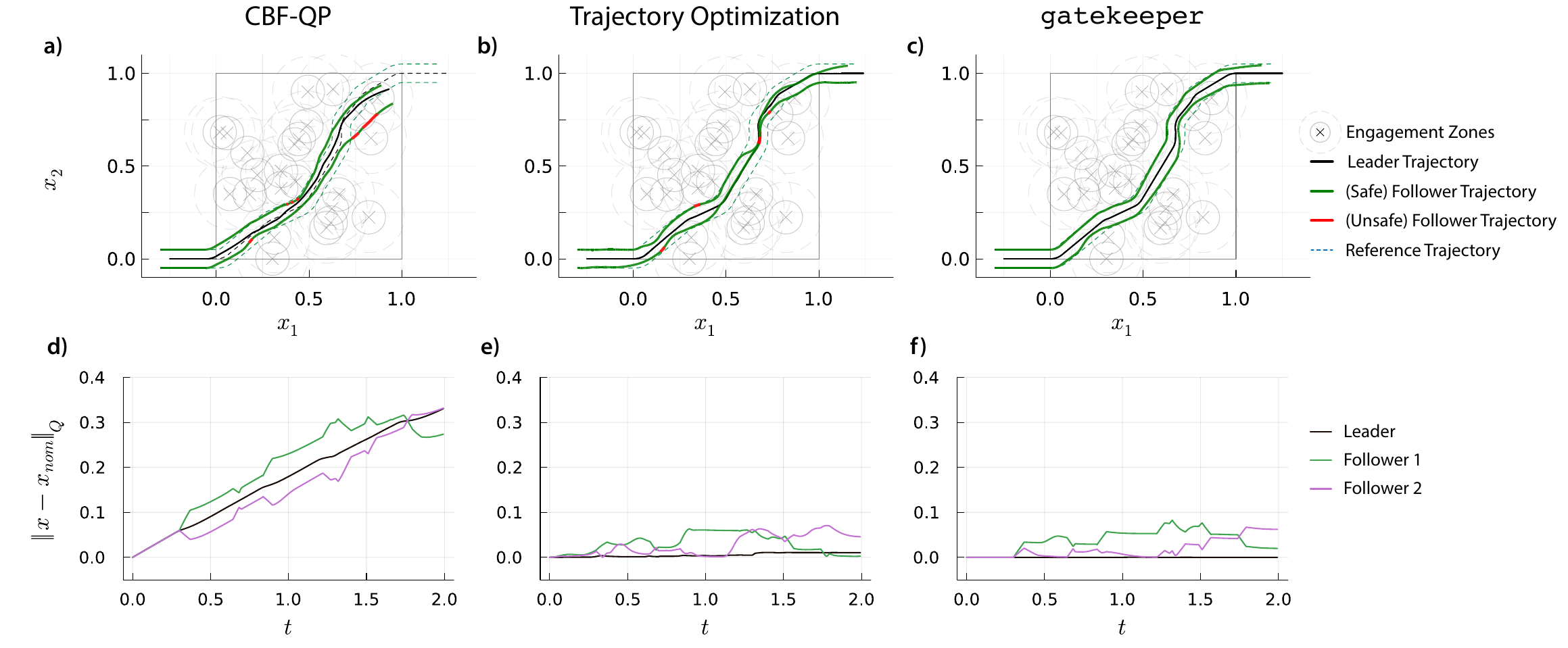}
    \caption{Simulation Results. (a, b, c)~depict the trajectories of the leader and the follower agents through the domain with 24~\acp{EZ} using the different planning methods. The leader's trajectory is drawn in black, while the follower's trajectories are drawn in green (when not in collision with a \ac{EZ}), and in red (when in collision). While \ac{CBF-QP} and IPOPT have safety violations, \gatekeeper{} does not. 
    (d, e, f)~show the deviation of the agents from the desired trajectory, with $Q = \rm{diag}(1, 1, 0)$. While the deviation in IPOPT and \gatekeeper{} are similar, the deviation is higher for the \ac{CBF-QP}. }
    \label{fig:results}
\end{figure*}

Comparing Figures \ref{fig:results}a, \ref{fig:results}b, \ref{fig:results}c, we can see closed-loop trajectories of the agents as they navigate through the environment. Only the \gatekeeper{} method shows no safety violations. For the \ac{CBF-QP} method, the $h_j$ functions may not be \acp{CBF}, especially under input constraints and the presence of multiple constraints, and therefore can lead to violations. The trajectory optimization method also has collisions, since there is no guarantee that IPOPT will find a feasible solution satisfying all the constraints. 

Figures \ref{fig:results}d, \ref{fig:results}e, \ref{fig:results}f compare the deviation of the trajectories from the desired trajectories. Here, we see that the \ac{CBF-QP} method has large deviations, since the controller reduces the $v$ to try to meet safety constraints. Furthermore, notice that the leader also deviates from the reference trajectory, even though the leaders path is safe. This is because even though the trajectory is safe, the trajectory approaches the boundary too quickly for the \ac{CBF} condition to be satisfied. Comparing Figures \ref{fig:results}e, \ref{fig:results}f, notice that the deviations in \gatekeeper{} and the trajectory optimization methods are comparable, indicating that \gatekeeper{} has minimal sub-optimality.

Finally, we compare the computation load. The total computation time required to generate the closed-loop trajectory all three agents is as follows: \ac{CBF-QP}: 9.49~s, trajectory optimization: 302.65~s,  \gatekeeper{}: 3.61~s. Ergo, \ac{CBF-QP} and \gatekeeper{} required only 3.1\% and 1.2\% of the computational time of nonlinear trajectory optimization, but only \gatekeeper{} resulted in safe trajectories. Our proposed approach is significantly computationally cheaper and has strong guarantees of constraint satisfaction, despite multiple state and input constraints. 





\section{CONCLUSIONS}

We have presented \gatekeeper{}, a flexible safety framework that guarantees safe execution in real-time planning systems by committing to a safety-verified trajectory with a known backup strategy. In particular, we quantify the suboptimality of the \gatekeeper{} approach, a quantity that we can compute in real-time. 

We demonstrated \gatekeeper{} on a challenging multi-agent formation flight task with tight safety margins and complex constraints. Compared to \ac{CBF}-based control and trajectory optimization, \gatekeeper{} was the only method to maintain safety throughout, while also being significantly faster and only minimally suboptimal.

These results suggest \gatekeeper{} is a promising direction for safety-critical robotics, especially when robustness and computational efficiency are paramount. Future work includes extending \gatekeeper{} to handle inter-agent collision constraints in a distributed communication network.

\bibliographystyle{IEEEtran}
\bibliography{references}

\end{document}